\documentclass[prb,aps,preprint,showpacs,preprintnumbers,
amsmath,amssymb,floatfix]{revtex4-1}
\usepackage[dvips]{color}	
\usepackage{graphicx}
\usepackage{amsmath, amsthm, amssymb}
\newtheorem{theorem}{Theorem}

\newtheorem{lemma}{Lemma}

\theoremstyle{remark}
\newtheorem*{remark}{Remark}
\newcommand{\wwto}{\stackrel{w}{\to}}

\newcommand{\supp}{\text{supp}  }

\newcounter{foo}

\begin{document}

\title{Bound States at Threshold Resulting from Coulomb Repulsion}
\author{Dmitry K. Gridnev}
\email[Electronic address:] {gridnev|at|fias.uni-frankfurt.de}
\affiliation{FIAS, Ruth-Moufang Strasse 1, D--60438 Frankfurt am Main, Germany}
\altaffiliation[On leave from:  ]{ Institute of Physics, St. Petersburg State
University, Ulyanovskaya 1, 198504 Russia}
\begin{abstract}
The eigenvalue absorption for a many--particle Hamiltonian depending on a
parameter is analyzed in the framework of
non--relativistic quantum mechanics. The long--range part of pair potentials is
assumed to be pure Coulomb and no restriction on the particle statistics is
imposed.
It is proved that if the lowest
dissociation threshold corresponds to the decay into two likewise non--zero
charged clusters then the bound state,
which approaches the threshold, does not spread and eventually
becomes the bound state at threshold. The obtained results have
applications in atomic and nuclear physics. In particular,  we prove that an
atomic
ion with the critical charge $Z_{cr}$ and $N_e$ electrons
has a bound state at threshold given that $Z_{cr} \in (N_e -2 , N_e -1)$,
whereby the electrons are treated as fermions and the mass of the nucleus
is finite.
\end{abstract}

\maketitle


\section{Introduction}\label{sec:1}

In Refs.~\onlinecite{1,2} it was proved that a critically bound N--body system, where none
of the subsystems has bound states with $E \leq 0$ and 
 particle pairs have no zero energy resonances,  
has a square integrable state at zero energy. The
condition on the absence of 2--body zero energy resonances was shown to be
essential in the three--body case \cite{1}. Here we consider the $N$--particle
system, 
where particles can be charged and apart from short--range pair--interactions
may also interact 
via Coulomb attraction/repulsion. The formation of bound states at threshold in
the two--particle case when the particles Coulomb repel each other is
well--studied \cite{3,gest}. In the three--particle case there is a well--known proof
\cite{ostenhof} that a two--electron ion with an infinitely heavy nucleus has a bound state at
threshold, when the nuclear charge becomes critical. 

Our aim here is to
investigate the general many--particle case. Here we generalize the result in
Ref.~\onlinecite{ostenhof} to the case of many electron ions with Fermi statistics
and finite nuclear mass. In the proofs we shall use the bounds on Green's
functions from Ref.~\onlinecite{3} as well as the technique of spreading sequences from
Ref.~\onlinecite{1}, 
that is we prove the eigenvalue absorption by demonstrating that the wave
functions corresponding to bound states do not spread, c.f. Theorem~1 in
Ref.~\onlinecite{1}. A different approach based on the calculus of variations was recently developed in Ref.~\onlinecite{frank}, where the authors 
give an alternative proof to the result in Ref.~\onlinecite{ostenhof}. The authors in Ref.~\onlinecite{frank} indicate that their approach could be 
generalized to the many--particle case. In the present paper as well as in Refs.~\onlinecite{ostenhof,frank} one uses essentially the same idea, namely, one uses the fact that 
the weak limit of ground state wave functions is a solution to the Schr\"odinger equation at the threshold. The hardest part is to prove that the weak limit is not 
identically zero. Our approach differs from the ones in Refs.~\onlinecite{ostenhof,frank} in that we use the upper  bounds on the two--particle Green's functions\cite{3}. 

The paper is organized as follows. In Sec.~\ref{sec:2} we introduce notations,
formulate the main theorem and prove a number of technical lemmas. 
In Sec.~\ref{sec:3} we derive an upper bound on the Green's function, which is
used in Sec.~\ref{sec:4} for the proof of Theorem~\ref{th:main}. 
In Sec.~\ref{sec:5} we discuss two main applications 
of Theorem~\ref{th:main} concerning the stability diagram of three Coulomb
charges (Theorem~\ref{th:martin} in Sec.~\ref{sec:5.1}) and negative atomic
ions 
(Theorem~~\ref{th:phys} in Sec.~\ref{sec:5.2}). In Appendix~\ref{sec:6} we
derive various criteria for non--spreading sequences. 

Let us mention physical applications. The effect when a
size of a bound system increases near the threshold and by far exceeds the
scales set by attractive parts of potentials was discovered in
neutron halos, helium dimer, Efimov states, for
discussion see Refs.~ \onlinecite{fedorov,zhukov,efimov,hansen0}. Here we
demonstrate that in a many--particle system similarly to the two--body case
\cite{3,gest} a Coulomb repulsion between possible decay
products blocks the spreading of bound states and forces an $L^2$ bound state at
threshold. In
nuclear physics, this, in particular, explains why contrary to neutron halos no
proton halos are found \cite{hansen}.

\section{Formulation of the Main Theorem}\label{sec:2}

We consider the $N$--particle Hamiltonian ($N \geq 3$)
\begin{gather}
    H(\lambda) = H_0 + V(\lambda) , \label{xc31}\\
    V(\lambda) := \sum_{1 \leq i<j \leq N} V_{ij} (\lambda) \equiv \sum_{1 \leq
i<j \leq N} \left[ U_{ij} (\lambda; r_i - r_j ) + \frac{q_i
(\lambda) q_j (\lambda)}{|r_i -r_j|}\right] , \label{:xc31}
\end{gather}
where $\lambda \in \mathbb{R}$ is a parameter, $H_0$ is the kinetic energy
operator with the center of mass
removed,  $r_i \in
\mathbb{R}^3$ are particles' position vectors and $q_i (\lambda) \in \mathbb{R}$
denote the particles' charges depending on $\lambda$. We shall assume that
$U_{ij} (\lambda; r) \in L^2 (\mathbb{R}^3) + L^\infty_\infty (\mathbb{R}^3)$
for each given 
$\lambda$. Here $L^\infty_\infty (\mathbb{R}^n)$ 
denotes the space of bounded Borel functions vanishing at infinity. 
We shall also take particle spins into account, though we shall consider only
spin--independent Hamiltonians. 
The Hamiltonian acts in 
$L^2 (\mathbb{R}^{3N-3}) \oplus L^2 (\mathbb{R}^{3N-3}) \oplus \cdots \oplus L^2
(\mathbb{R}^{3N-3}) \equiv L^2 ( \mathbb{R}^{3N-3}; \mathbb{C}^{n_s})$, 
where the direct sum has $n_s = (2s_1 + 1) (2s_2 + 1)\ldots (2s_N + 1)$
summands and $s_i$ denotes the spin of particle $i$. 
Similar notation for the Hilbert space can be found in Refs.~\onlinecite{quotesimon,quotethaller}. 
By Kato's theorem \cite{reed,teschl} $H(\lambda)$ is self--adjoint on 
$D(H_0) = \mathcal{H}^2 (\mathbb{R}^{3N-3}; \mathbb{C}^{n_s}) \subset L^2 ( \mathbb{R}^{3N-3}; \mathbb{C}^{n_s}) $, where 
 $\mathcal{H}^2 (\mathbb{R}^{3N-3}; \mathbb{C}^{n_s}) \equiv \mathcal{H}^2 (\mathbb{R}^{3N-3}) \oplus \cdots \oplus \mathcal{H}^2 (\mathbb{R}^{3N-3})$ and
$\mathcal{H}^2 (\mathbb{R}^{3N-3})$ denotes the 
corresponding Sobolev space \cite{teschl,liebloss}. A function $f \in L^2 (
\mathbb{R}^{3N-3}; \mathbb{C}^{n_s})$ depends explicitly on the arguments as 
$f (x, \sigma_1 , \ldots, \sigma_N)$, where $x \in \mathbb{R}^{3N-3}$ and $\sigma_i
\in \{ \frac{s_i}2 , \frac{s_i}2 - 1, \ldots , - \frac{s_i}2 \}$ are the spin
variables.

We treat the particles with integer spins as bosons and particles with
half--integer spin as fermions. 
$\mathcal{P}$ denotes the orthogonal projection operator on the 
subspace of functions, which are symmetric with respect to
the interchange of bosons and antisymmetric with respect to the interchange of
fermions. We denote the
bottom of the continuous spectrum by 
\begin{equation}
 E_{thr}(\lambda) := \inf \sigma_{ess} \, H(\lambda) \mathcal{P} . 
\end{equation}

We shall use the function  $\eta_\alpha \colon
\mathbb{R}^n \to \mathbb{R}$, which determines the asymptotic behavior at
infinity
\begin{equation}\label{xweta}
 \eta_\alpha (r) := \chi_{\{r|\; |r| \leq 1\}} + \chi_{\{r|\; |r| > 1\}}
|r|^{\alpha} , 
\end{equation}
where $r \in \mathbb{R}^n, \alpha \in \mathbb{R}_+$ and $\chi_A$ always denotes
the
characteristic function of the set $A$.
Note that $\eta_\alpha (r)$ is continuous and $\eta_{\alpha_1} \eta_{\alpha_2} =
\eta_{\alpha_1 \alpha_2}$. We make the following assumptions
\begin{list}{R\arabic{foo}}
{\usecounter{foo}
    \setlength{\rightmargin}{\leftmargin}}
\item $H(\lambda)$ is defined for an infinite sequence of parameter values
$\lambda_1, \lambda_2, \ldots $ and $\lambda_{cr}$, where $\lim_{n\to \infty}
\lambda_n = \lambda_{cr}$. For all $ \lambda_n $ there is $E(\lambda_n) \in
\mathbb{R}, \psi_n \in D(H_0)$
such that $H(\lambda_n) \psi_n = E(\lambda_n) \psi_n$, where $\| \psi_n \| = 1$,
 $\mathcal{P}\psi_n = \psi_n$
and $E(\lambda_n) < E_{thr}(\lambda_n)$. Besides, $\lim_{n \to \infty}
E(\lambda_n) = \lim_{n \to \infty} E_{thr}(\lambda_n) = E_{thr}
(\lambda_{cr})$.

\item $\sup_{\lambda=\lambda_n, \lambda_{cr}}\; |U_{ij} (\lambda; y )| \leq
\tilde U (y)$ and  $\sup_{\lambda=\lambda_n, \lambda_{cr}}\; |q_i(\lambda)
q_j(\lambda)| \leq q_0$, where $\tilde
U(y)$ is such that
$\eta_\delta (y) \tilde U(y) \in L^2(\mathbb{R}^3) +
L^\infty_\infty(\mathbb{R}^3)$ and $\delta \in (3/2, 2)$, $q_0 \in (0, \infty)$
are fixed constants.
Additionally, $\lim_{n \to \infty} \bigl\| \bigl[ V(\lambda_n) - V(\lambda_{cr})
\bigr] f \bigr\| =
0$ for all $f \in C^\infty_0 (\mathbb{R}^{3N-3} )$.
\end{list}

Let $a = 1, 2,
\ldots, (2^{N-1}-1)$ label all the distinct ways \cite{ims} of partitioning
particles into two non--empty clusters $\mathfrak{C}^{a}_1$ and
$\mathfrak{C}^{a}_2$. We define the Jacobi intercluster coordinates for the
clusters $\mathfrak{C}^{a}_{1,2}$ as 
$x^{a,1}_{i}$ and $x^{a,2}_{j}$ respectively, where $i = 1,2, \ldots,
(\#\mathfrak{C}^{a}_1 - 1)$ and
$j = 1,2, \ldots, (\#\mathfrak{C}^{a}_2 - 1)$ (the symbol $\#$ denotes the
number
of particles in the corresponding cluster).
By $x^a$ we denote the full set of intercluster
coordinates and we set 
\begin{equation}\label{xiamod}
    |x^a| = \sum_{i=1}^{\# \mathfrak{C}^{a}_1 - 1} |x^{a,1}_{i}| +
\sum_{j=1}^{\#
\mathfrak{C}^{a}_2 - 1} |x^{a,2}_{j}| . 
\end{equation}
$R_a$ points from the center
of mass
of $\mathfrak{C}^{a}_1$ to the center of mass of $\mathfrak{C}^{a}_2$. The full
set of Jacobi coordinates
is $(x^a , R_a) \in \mathbb{R}^{3N-3}$.

We denote the sum of interaction cross terms between the clusters by
\begin{equation}\label{14ia}
I_a (\lambda) := \sum_{\substack{i \in \mathfrak{C}^{a}_1 \\j
\in\mathfrak{C}^{a}_2}}
V_{ij} (\lambda) . 
\end{equation}
The product of net charges of the clusters is defined as
\begin{equation}\label{netcharge}
    Q^a (\lambda) := \sum_{i \in \mathfrak{C}^{a}_1} \sum_{j \in
\mathfrak{C}^{a}_2}
q_i (\lambda) q_j (\lambda) . 
\end{equation}

The projection operators on the proper symmetry subspace for the particles
within clusters $\mathfrak{C}_1^{(a)}$ and $\mathfrak{C}_2^{(a)}$ are
$\mathcal{P}_1^{(a)}$ and $\mathcal{P}_2^{(a)}$ respectively. Namely,
$\mathcal{P}_i^{(a)}$ projects on a subspace of functions, 
which are antisymmetric with respect to the 
interchange of fermions in $\mathfrak{C}_i^{(a)}$ and symmetric with respect to
the interchange of bosons in $\mathfrak{C}_i^{(a)}$ ($i = 1, 2$). 
Naturally,
$\mathcal{P} \mathcal{P}^{(a)}_{1,2} = \mathcal{P}^{(a)}_{1,2} \mathcal{P} =
\mathcal{P}$ and $[\mathcal{P}_1^{(a)} ,\mathcal{P}_2^{(a)}] = 0$. We also
define $\mathcal{P}^{(a)} := \mathcal{P}_1^{(a)} \mathcal{P}_2^{(a)}$. The
Hamiltonian (\ref{xc31}) can be decomposed in the following way
\begin{equation}\label{hthra}
    H(\lambda) = H^{(a)}_{thr} (\lambda) - \frac{\hbar^2}{2\mu_a} \Delta_{R_a} +
I_a (\lambda),
\end{equation}
where $H^{(a)}_{thr} (\lambda)$ is the Hamiltonian of the clusters' intrinsic
motion
and $\mu_a$ denotes the reduced mass derived from clusters' total masses. From
now on without loss of generality we set $\hbar^2/(2\mu_a) = 1$.

It is convenient to treat the Hilbert space as the tensor product
$L^2 (\mathbb{R}^{3N-3}; \mathbb{C}^{n_s}) = L^2
(\mathbb{R}^{3N-6}; \mathbb{C}^{n_s}) \otimes L^2 (\mathbb{R}^{3})$, where the
first term in the product
corresponds to the space associated with $x^a$ coordinates and spin variables, while the second one
refers to the space
associated with the $R_a$ coordinate. In such case the operator $H^{(a)}_{thr}$
has the form
$H^{(a)}_{thr} = H^{a}_{thr} \otimes 1$, where $H^{a}_{thr}$ is the restriction
of $H^{(a)}_{thr}$ to $L^2 (\mathbb{R}^{3N-6}; \mathbb{C}^{n_s})$. The
coordinate $R_a$ is unaffected by permutations of particles 
within the clusters $\mathfrak{C}_1^{(a)}$ or 
$\mathfrak{C}_2^{(a)}$. Therefore, $\mathcal{P}^{(a)} =
\mathcal{P}^{a} \otimes 1$, where $\mathcal{P}^{a}$ denotes  the restriction of
$\mathcal{P}^{(a)}$ to the space associated with $x^a$ coordinates and spin variables.

The set of assumptions is continued as follows. 
\begin{list}{R\arabic{foo}}
{\usecounter{foo}
    \setlength{\rightmargin}{\leftmargin}}
\setcounter{foo}{2}
\item
For $\lambda = \lambda_n , \lambda_{cr}$ and $a=1,\ldots,\mathfrak{N}$ one has
$\inf \sigma
\bigl(H^a_{thr}(\lambda) \mathcal{P}^a\bigr) = E_{thr}(\lambda)$. There is
$|\Delta \epsilon| >0$  such that the following
inequalities hold for $\lambda=\lambda_n, \lambda_{cr}$
\begin{gather}
\inf \sigma_{ess} \bigl(H^a_{thr}(\lambda) \mathcal{P}^a\bigr) \geq
E_{thr}(\lambda) + 2
|\Delta \epsilon| \quad (a=1,\ldots, \mathfrak{N}) , \label{need1}\\
\Bigl[ H^a_{thr} (\lambda) - E_{thr}(\lambda)\Bigr] \mathcal{P}^a \geq |\Delta
\epsilon|
\mathcal{P}^a \quad (a=\mathfrak{N}+1,\ldots,2^{N-1}-1) .\label{need2}
\end{gather}
\end{list}
The requirement R3 says that the bottom of the continuous spectrum of
$H(\lambda)$ is
set by the decomposition into those two clusters that correspond to any of the
decompositions $a=1, \ldots, \mathfrak{N}$. Inequality (\ref{need1}) introduces
a gap
between the ground state energy of the two clusters and other excited states.
For
$a=1,\ldots, \mathfrak{N}$ and $\lambda = \lambda_n, \lambda_{cr}$ we define the
projection operator
acting on $L^2 (\mathbb{R}^{3N-6}; \mathbb{C}^{n_s}) $
\begin{equation}\label{projopa}
    P^a_{thr} (\lambda) = \mathbb{P}^a_{[E_{thr}(\lambda) , E_{thr}(\lambda) +
|\Delta
\epsilon|]}\; ,
\end{equation}
where $\{\mathbb{P}^a_{\Omega}\}$ are spectral projections of
$H^a_{thr}(\lambda)
\mathcal{P}^a$. Note that by R3 the projection operators $P^a_{thr}
(\lambda_n)$, $P^a_{thr} (\lambda_{cr})$
have a finite dimensional range. 

The last assumption introduces the uniform control over the fall off of clusters' wave
functions
\begin{list}{R\arabic{foo}}
{\usecounter{foo}
    \setlength{\rightmargin}{\leftmargin}}
\setcounter{foo}{3}
\item
There are constants $A, \beta
>0$ such that
\begin{equation}\label{perxp}
    \bigl\| e^{\beta |x^a|} P^a_{thr}(\lambda) \bigr\| \leq A
\end{equation}
for $\lambda = \lambda_n , \lambda_{cr}$ and $a=1,2, \ldots, \mathfrak{N}$.
\end{list}
Due to R3 there must exist orthonormal $\varphi^a_i (\lambda)\in D(-\Delta)
\subset L^2(\mathbb{R}^{3N-6}; \mathbb{C}^{n_s})$ for $i=1,2,\ldots, n_a
(\lambda)$ such that
\begin{equation}
 P_{thr}^a (\lambda) = \sum_{i=1}^{n_a (\lambda)} E_i^a (\lambda) \varphi^a_i
(\lambda) \bigl(\cdot, \varphi^a_i (\lambda) \bigr) \quad \textrm{for $a = 1,
\ldots, \mathfrak{N}$ and $\lambda = \lambda_n , \lambda_{cr}$} ,
\end{equation}
where $E_i^a (\lambda) \in [E_{thr}(\lambda), E_{thr}(\lambda) +
|\Delta\epsilon|]$. Note that $H^a_{thr}(\lambda_n) \varphi^a_i (\lambda_n) =
E_i^a (\lambda_n) \varphi^a_i (\lambda_n)$, therefore,
$\|-\Delta \varphi^a_i (\lambda_n)\|$ is uniformly bounded, c.~f. Lemma~1 in
Ref.~\onlinecite{1}. Applying Lemma~\ref{lem:4} below  and using R4 we conclude that
there exists an integer $\omega$ such that $n_a (\lambda_{cr}), n_a
(\lambda_n) \leq \omega$.
\begin{lemma}\label{lem:4}
 Suppose that the orthonormal set of function $\phi_1 , \ldots,
\phi_\mathcal{N} \in D(-\Delta) \subset L^2(\mathbb{R}^d ; \mathbb{C}^{n_s})$ is
such that
$\| - \Delta \phi_i \| \leq T$ and $\| e^{\beta|x|}\phi_i \|\leq A$ for $i = 1,
\ldots, \mathcal{N}$, where $T, A, \beta >0$ are constants. If $d \geq 3$ then
\begin{equation}\label{lbcwros}
 \mathcal{N} \leq C_d \frac{(2 T)^{d/2} |\ln 2A|^d}{(2 \beta)^d} n_s ,
\end{equation}
where $C_d$ is the Lieb's constant in the Cwikel--Lieb--Rosenbljum bound.
\end{lemma}
\begin{proof}
From $\| e^{\beta|x|}\phi_i \|\leq A$ it follows that
\begin{equation}\label{22a}
(\phi_i , \chi_{\{x |\; |x| \leq R \}} \phi_i )  \geq \frac 12 ,
\end{equation}
where we set $R:= (\ln{2A})/(2\beta)$. Hence,
\begin{equation}\label{24}
\Bigl( \phi_i , \left[ - \Delta - 2T\chi_{\{x |\; |x| \leq R \}}
\right]\phi_i \Bigr)  < 0 .
\end{equation}
By the min--max principle $\mathcal{N}$ does not exceed the number of negative
energy bound states of the operator in square brackets in (\ref{24}). This
number, in turn, is equal to the number of negative energy bound states of the 
operator in square brackets considered in $L^2 (\mathbb{R}^d)$ times $n_s$ due
to the spin degeneracy. Now
(\ref{lbcwros})  follows from the Cwikel--Lieb--Rosebljum bound \cite{reed,cwikel}.
\end{proof}
Now we can formulate the main theorem.
\begin{theorem}\label{th:main}
Suppose that $H(\lambda)$ satisfies $R1-R4$ and
\begin{equation}\label{Q0new}
Q_0 := \inf_{a=1, \ldots, \mathfrak{N}} \inf_{\lambda=\lambda_n, \lambda_{cr}}
Q^a (\lambda) >  0 . 
\end{equation}
Then the sequence $\psi_n$ does not spread and there exists $\psi_{cr} \in D(H_0) \subset L^2
(\mathbb{R}^{3N-3};\mathbb{C}^{n_s})$ such that
$H(\lambda_{cr}) \psi_{cr} = E_{thr} (\lambda_{cr}) \psi_{cr}$, where
$\| \psi_{cr} \| = 1$ and $\psi_{cr} = \mathcal{P} \psi_{cr}$.
\end{theorem}
Let us remark that the term spreading was defined in Ref.~\onlinecite{1} for sequences in
$L^2 (\mathbb{R}^d)$. We shall say that a sequence 
$f_n \in L^2 (\mathbb{R}^d ; \mathbb{C}^{n_s})$ spreads if $f_n (x, \sigma_1 ,
\ldots, \sigma_{N})$ spreads for all possible fixed values of the spin
variables.  
We postpone the proof of Theorem~\ref{th:main} to Sec.~\ref{sec:4}.

Together with the upper bound on the Green's function derived in the next
section the following lemma is the
key ingredient in the proof of Theorem~\ref{th:main}.
\begin{lemma}\label{lem:thet}
There is $\Theta_a (x) \in L^2 (\mathbb{R}^{3N-3}) +
L^\infty_\infty(\mathbb{R}^{3N-3}) $ independent of $\lambda$ such that 
\begin{equation}\label{thetaa}
    \Bigl| e^{-\beta  |x^a|} \eta_\delta (R_a) \bigl[ I_a (\lambda) - Q^a
(\lambda)\eta_{-1} (R_a) \bigr] \Bigr|\leq \Theta_a (x) 
\end{equation}
for $\lambda= \lambda_n, \lambda_{cr}$ defined in R1, $\delta$ defined in R2 and
$\beta$ defined in R4. 
\end{lemma}

\begin{proof}
The statement of the lemma is based on the following inequality, which can be
checked directly. For all $s, s' \in \mathbb{R}^3$
\begin{equation}\label{paxi}
    \left| \frac{\chi_{\{s,s' |\; |s-s'| \geq 1\}}}{|s-s'|} -
\eta_{-1}(s)\right| \leq 2 \eta_{2}(s')\eta_{-2}(s) . 
\end{equation}
For fixed $s'$ the term on the lhs of (\ref{paxi}) falls off like $|s|^{-2}$. We
write
\begin{gather}
\Bigl| I_a (\lambda) - Q^a (\lambda)\eta_{-1} (R_a) \Bigr|  \leq
\sum_{\substack{i \in
\mathfrak{C}^{a}_1 \\j \in\mathfrak{C}^{a}_2}} \tilde U_{ij}  +
\sum_{\substack{i \in \mathfrak{C}^{a}_1 \\j \in\mathfrak{C}^{a}_2}}
\frac{q_0}{|r_i - r_j|} \chi_{\{x |\; |r_i-r_j| \leq 1\}} \nonumber\\
+ \sum_{\substack{i \in \mathfrak{C}^{a}_1 \\j \in\mathfrak{C}^{a}_2}} \bigl|
q_i (\lambda) q_j (\lambda) \bigr| \left| \frac{\chi_{\{x |\; |r_i-r_j| \geq
1\}}}{|r_i-r_j|} - \eta_{-1} (R_a) \right| . \label{xwsehz}
\end{gather}
For any cluster decomposition $a$ and $i \in \mathfrak{C}^{a}_1 , j \in
\mathfrak{C}^{a}_2 $
\begin{equation}
r_j - r_i = R_a + \sum_{i=1}^{\#\mathfrak{C}^{a}_1 - 1} c_i^{a,1} x_i^{a,1} +
\sum_{i=1}^{\#\mathfrak{C}^{a}_2 - 1} c_i^{a,2} x_i^{a,2} , 
\end{equation}
where $c_i^{a,1}$, $c_i^{a,2}$ are numerical coefficients depending on masses.
It is easy to see that the coefficient in front of $R_a$ is always $1$ by fixing
$|x^a|$ and
taking $|R_a| \gg 1$. Therefore, by (\ref{paxi}) we have
\begin{equation}\label{pshe}
    \left| \frac{\chi_{\{x |\; |r_i-r_j| \geq 1\}}}{|r_i-r_j|} - \eta_{-1} (R_a)
\right| \leq c_0 \eta_2 (|x^a|) \eta_{-2}(R_a) ,
\end{equation}
where $c_0 > 0$ is some constant. Substituting (\ref{pshe}) into (\ref{xwsehz})
we conclude that the inequality (\ref{thetaa}) would be true if we set $\Theta_a
= \Theta_{a1} + \Theta_{a2}$, where
\begin{gather}
\Theta_{a1}(x) := e^{-\beta  |x^a|} \eta_\delta (R_a) \sum_{\substack{i \in
\mathfrak{C}^{a}_1 \\j \in\mathfrak{C}^{a}_2}} \Bigl[ \tilde U_{ij} +
\frac{q_0}{|r_i - r_j|} \chi_{\{x |\; |r_i-r_j| \leq 1\}} \Bigr] , \\
\Theta_{a2}(x) := c_0 N (N-1) q_0 e^{-\beta  |x^a|} \eta_2 (|x^a|)
\eta_{\delta-2}(R_a) . 
\end{gather}
Using R2 it is easy to see that $\Theta_{a1} \in L^2 (\mathbb{R}^{3N-3}) +
L^\infty_\infty(\mathbb{R}^{3N-3})$. Because $\delta < 2$ we have $\Theta_{a2}
\in L^\infty_\infty(\mathbb{R}^{3N-3})$.
\end{proof}

\section{Upper Bound on the Two Particle Green's Function}\label{sec:3}

Consider the following integral operator on $L^2 (\mathbb{R}^3)$
\begin{equation}\label{def:Gk}
    G^{c}_k (A) = \left[ -\Delta + A \eta_{-1}(r) + k^2\right]^{-1} ,
\end{equation}
for $A,k >0$, whose integral kernel we denote as $G^{c}_k (A;r,r')$ (the
superscript ``c'' refers to
``Coulomb''). Note that $G^{c}_k (A;r,r') \leq G^{c}_{\tilde k} (\tilde A
;r,r')$ away from $r = r'$ if either 
$\tilde A \leq A$ or $\tilde k \leq k$, c.~f. Lemma~1 in Ref.~\onlinecite{3}. The following
Lemma uses the upper bound on a two particle Green's
function from Ref.~\onlinecite{3}.
\begin{lemma}\label{lemma:1}
There is $b(A) >0$ such that for all $A>0$, $n>0$
\begin{equation}\label{mainnorm}
    \sup_{k>0} \bigl\| G^{c}_k (A) \chi_{\{ r|\: |r| \leq n\}} \bigr\| \leq b(A)
n,
\end{equation}
where the norm on the lhs is the operator norm.
\end{lemma}
\begin{proof}
The operator $G^{c}_k (A) $ is an integral operator with a positive kernel
\cite{lpestim} and, hence, it suffices to consider (\ref{mainnorm}) for $n>1$.
For a shorter notation we denote $\chi_n := \chi_{\{r|\; |r| \leq n\}}$.
Obviously
\begin{gather}
\| G^{c}_{k}(A)\chi_n \| \leq \| \chi_{4n} G^{c}_{k} (A) \chi_n \| + \|
(1-\chi_{4n}) G^{c}_{k}(A)  \chi_{n} \| \nonumber \\
\leq \|\chi_{4n} G^{c}_{k} (A)\chi_{4n} \|  + \| (1-\chi_{4n}) G^{c}_{k}(A) 
\chi_{n}
\| \label{1} ,
\end{gather}
where the last inequality follows from $G^{c}_{k} (A)$ being an integral
operator
with a positive kernel \cite{3,lpestim}. We shall derive the following estimates
$\|
\chi_{4n} G^{c}_{k} (A) \chi_{4n} \| = \mathcal{O}(n)$ and
$\| (1-\chi_{4n}) G^{c}_{k} (A) \chi_n \| = \hbox{o}(n)$ for $n \to \infty$,
from which
the
the statement of the Lemma follows. The first term on the rhs of (\ref{1}) is
the norm of the self--adjoint operator, which can be estimated as follows
\begin{gather}
\| \chi_{4n} G^{c}_{k} (A)\chi_{4n} \| = \sup_{\| f \| =
1} \Bigl( \chi_{4n} f,  G^{c}_{k} (A)\chi_{4n} f \Bigr) \leq \nonumber \\
\sup_{\| f \| = 1} \Bigl( \chi_{4n} f, (A\eta_{-1})^{-1}\chi_{4n} f \Bigr)  = 4
A^{-1} n , \label{4}
\end{gather}
where we have used the inequality $(B+ \varepsilon)^{-1} \leq
(C+ \varepsilon)^{-1} $ for non--negative self--adjoint operators $B \geq C \geq 0$ and $\varepsilon >0$ (see, for example Ref.~\onlinecite{glimmjaffe}, Proposition
A.2.5 on page 131). Thus $\| \chi_{4n}
G^{c}_{k} (A) \chi_{4n}\| = \mathcal{O}(n)$ as claimed.

Let us now consider the second term on the rhs of (\ref{1}). We
shall need the bound on the Green's function from Ref.~\onlinecite{3}. Let $\tilde G_k (a;
r,r')$ denote the integral kernel of the following operator on $L^2
(\mathbb{R}^3)$
\begin{equation}\label{5}
\tilde G_k (a) = \left[ -\Delta + \left( \frac{a^2}4 |r|^{-1}  +
\frac{a}4 |r|^{-3/2} \right)\chi_{\{r|\: |r| \geq 1\}} + k^2
\right]^{-1} . 
\end{equation}
Lets us set $a$ equal to the positive root of the equation $a(a+1) = 4A$. Then
we get
\begin{equation}\label{.9}
A \eta_{-1} (r) \geq \left( \frac{a^2}4 |r|^{-1}  + \frac a4|r|^{-3/2}
\right)\chi_{\{r|\: |r| \geq 1\}},
\end{equation}
which means that $G^{c}_{k} (A;r,r' ) \leq \tilde G_k (a;r,r' )$ pointwise for
all $r \neq r' $, see Ref.~ \onlinecite{3}. The upper bound on $\tilde G_k (a; r,r' )$ from
Ref.~\onlinecite{3} (Eqs.(42)--(43) and Eqs.~(39)--(40) in Ref.~\onlinecite{3}) reads
\begin{equation}\label{7}
\tilde G_k (a;r,r') \leq \frac 1{4 \pi |r-r'|} \times
\left\{
\begin{array}{ll}
    1  & \quad \textrm{for $|r-r'| \leq \tilde R_0 $} \\
    \exp\left\{ \tilde a \sqrt{\tilde R_0 } - \tilde a \sqrt{|r-r'|} \right\} &
\quad \textrm{for $|r-r'| > \tilde R_0 $}, \\
    \end{array}
    \right.
\end{equation}
where $\tilde R_0 , \tilde a$ have to be chosen to satisfy the following
inequalities 
\begin{gather}
\tilde R_0 \geq 1 + |r'| , \label{12}\\
\tilde a \leq a \tilde R^{3/2}_{0} (\tilde R_0 + |r'|)^{-3/2} . \label{13}
\end{gather}
From the inequality (\ref{7}) we obtain the bound
\begin{equation}\label{8}
\tilde G_{k} (a;r,r') \chi_{\{ |r'| \leq n \}}\leq \frac 1{4
\pi |r-r'|} \times \left\{
\begin{array}{ll}
    1  & \quad \textrm{for $|r-r'| \leq 2n $} \\
    \exp\left\{ \frac a2 \bigl( \sqrt{2n } - \sqrt{|r -r'|} \bigr) \right\} &
\quad \textrm{for $|r-r'| > 2n $}, \\
    \end{array}
    \right.
\end{equation}
where we have set $\tilde R_0 = 2n $ and $\tilde a = a/2$. It is straightforward
to check that this choice of $\tilde R_0, \tilde a$ indeed satisfies
(\ref{12})--(\ref{13}). Taking into account that $G^{c}_{k} (A;r,r' ) \leq
\tilde G_k (a;r,r' )$ we finally get from (\ref{8}) the required bound
\begin{equation}\label{15}
G^{c}_{k} (A;r,r') \chi_{\{ r,r'|\: |r| \geq 4n ,  |r'| \leq n \}} \leq
\frac{e^{\frac a2 \left( \sqrt{2n} - \sqrt{|r|- n} \right)}}{4\pi (3n)} \chi_{\{
r,r'|\: |r| \geq 4n ,  |r'| \leq n \}} . 
\end{equation}
Note that the rhs of (\ref{15}) does not depend on $k$. Using the upper bound
(\ref{15}) and estimating the operator norm through the Hilbert--Schmidt norm we
get
\begin{gather}
    \| (1-\chi_{4n}) G^{c}_{k} (A)\chi_n \|^2 \leq \int_{|r|\geq 4n} dr
    \int_{|r'| \leq n} dr' |G^{c}_{k} (A;r,r' )|^2  \nonumber \\
    \leq \frac n{27} e^{a\sqrt{2n}} \int_{4n}^{\infty} e^{-a\sqrt{t-n}} t^2 dt . 
\label{18}
\end{gather}
The integral in (\ref{18}) can be calculated explicitly and we obtain $\|
(1-\chi_{4n}) G^{c}_{k} (A) \chi_n \| = \hbox{o}(n)$ as claimed.  \end{proof}

We shall need the following corollary of Lemma~\ref{lemma:1}
\begin{lemma}\label{lemma:2}
For fixed $A>0, \alpha >3/2$ the following inequality holds
\begin{equation}\label{20}
    \sup_{k>0} \bigl\|  G^{c}_{k} (A)\; \eta_{-\alpha}\bigr\| < \infty . 
\end{equation}
\end{lemma}
\begin{proof}
For an arbitrary $f \in L^{2} (\mathbb{R}^3 )$ we have
\begin{gather}
\bigl\| G^{c}_{k} (A)\eta_{-\alpha} f \bigr\|  = \lim_{N \to \infty} \Bigl\|
\sum_{n =
1}^N G^{c}_{k}  (A) \eta_{-\alpha}
(\chi_n - \chi_{n-1} ) f \Bigr\| \nonumber   \\
\leq \lim_{N\to \infty} \sum_{n = 1}^N \Bigl\| G^{c}_{k}  (A) \chi_n
\eta_{-\alpha}
(\chi_n - \chi_{n-1})^2 f \Bigr\| , 
 \label{32}
\end{gather}
where we have used $(\chi_n - \chi_{n-1})^2 = (\chi_n - \chi_{n-1})$
and $\chi_n (\chi_n - \chi_{n-1}) = (\chi_n - \chi_{n-1})$.
For the operator norms we have $\| \eta_{-\alpha} \chi_1 \| = 1$ and $\|
\eta_{-\alpha} (\chi_n - \chi_{n-1}) \| = (n-1)^{-\alpha}$ for
$n\geq 2$. Substituting these into (\ref{32}) and using Lemma~\ref{lemma:1} we
rewrite (\ref{32}) as
\begin{gather}\label{33}
\| G^{c}_{k} (A)\eta_{-\alpha} f \| \leq b(A) \lim_{N\to \infty} \left(
\bigl\|\chi_1 f \bigr\| + \sum_{n =
2}^N  n(n-1)^{- \alpha} \bigl\| (\chi_n - \chi_{n-1})f \bigr\|\right) . 
\end{gather}
Now using that $\sum_n \|(\chi_n - \chi_{n-1} )f \|^2 =\|f\|^2$ and applying the
Cauchy-Schwartz inequality we get from Eq.~(\ref{33})
\begin{gather}\label{34}
\| G^{c}_{k} (A)\eta_{-\alpha} \| \leq b(A) \left( 1 + \sum_{n =
2}^\infty  n^2 (n-1)^{- 2\alpha} \right)^{1/2} . 
\end{gather}
For $\alpha > 3/2$ the series on the rhs of Eq.~(\ref{34}) obviously
converge. \end{proof}

\section{Proof of the Main Theorem}\label{sec:4}

We shall need an analogue of the IMS localization formula, see Ref.~\onlinecite{ims}. The
functions $J_a
\in C^\infty 
(\mathbb{R}^{3N-3})$ form the partition
of unity $\sum_a J^{2}_a =1$ and are homogeneous of
degree zero in the exterior of the unit sphere, {\em i.e.} $J_a (\lambda x) =
J_a (x)$ for $\lambda \geq 1$, $|x|= 1$ (this makes $|\nabla
J_a |$ fall off at infinity). Additionally, there exists a constant $C > 0$ such
that
\begin{equation}\label{ims5}
    \supp J_a \cap \{ x | |x| > 1 \} \subset \{x|\; |x_i - x_j | \geq C |x|
\quad
    \textrm{for $i \in \mathfrak{C}_1^a , j \in \mathfrak{C}_2^a$}\} . 
\end{equation}
The functions of the IMS decomposition can be chosen
\cite{sigalsays1,sigalsays2}
invariant under permutations of particle coordinates both in
$\mathfrak{C}_1^a$ and in $\mathfrak{C}_2^a$, hence $[J_a, \mathcal{P}^{(a)}] =
0$. Note also that for all $f \in \mathcal{H}^2 (\mathbb{R}^{3N-3})$ one has $J_a f \in \mathcal{H}^2 (\mathbb{R}^{3N-3})$, c. f. Lemma~7.4 in Ref.~\onlinecite{liebloss} 
(the proof in Ref.~\onlinecite{liebloss} easily extends to the case of 
Sobolev spaces of higher order). 
The following version of the IMS localization formula can be verified by the
direct substitution
\begin{equation}\label{14sex}
    \Delta = \Delta \sum_{a,b} J^2_a J^2_b  = \sum_{a,b} J_a J_b \Delta J_b J_a
+ 2 \sum_a |\nabla J_a|^2 ,
\end{equation}
where $\Delta$ is the Laplace on $\mathbb{R}^{3N-3}$. Rescaling (\ref{14sex}) we
get
\begin{equation}\label{14}
    H_0 = \sum_{a,b} J_a J_b H_0 J_b J_a + 2 \sum_a \sum_{s=1}^{3N-3}
\mathfrak{m}_s |\partial_s J_a|^2 ,
\end{equation}
where $\mathfrak{m}_s$ are real coefficients depending on masses and the second
term
is relatively $H_0$ compact \cite{ims}.
We introduce
\begin{equation}\label{cycha}
    H_a (\lambda)= H(\lambda) - I_a (\lambda)= H^{(a)}_{thr}(\lambda) -
\Delta_{R_a} . 
\end{equation}
From (\ref{14}) it follows that
\begin{equation}\label{ims3}
    H(\lambda) = \sum_a J^2_a H_a (\lambda) J^2_a + \sum_{a\neq b} J_a J_b
H_{ab}(\lambda) J_b J_a
+ K(\lambda) , 
\end{equation}
where we define
\begin{gather}
    K(\lambda) := \sum_{a\neq b} J^2_a J_b^2 I_{ab}(\lambda) +  \sum_a \Bigl[
J^{4}_a I_a(\lambda)
+ 2 \sum_{s=1}^{3N-3} \mathfrak{m}_s |\partial_s J_a|^2  \Bigr] , \label{14k}\\
H_{ab}(\lambda) := H_0 + \sum_{s=1}^2 \sum_{p=1}^2 \: \sum_{i,j \in
\mathfrak{C}^{a}_s
\cap \mathfrak{C}^{b}_p} V_{ij} (\lambda) , \\
I_{ab} (\lambda) := H(\lambda) - H_{ab}(\lambda) . 
\end{gather}
The Hamiltonian $H_{ab}$ defined for $a \neq b$ contains intercluster
interactions of the following
four clusters $\mathfrak{C}^a_1 \cap \mathfrak{C}^b_1$, $\mathfrak{C}^a_1 \cap
\mathfrak{C}^b_2$, $\mathfrak{C}^a_2 \cap \mathfrak{C}^b_1$, $\mathfrak{C}^a_2
\cap \mathfrak{C}^b_2$,
while all interaction cross--terms between these four clusters are contained in
$I_{ab}$. (For
some partitions it might happen that one of the four clusters is empty). If we
define by $\mathcal{P}^{(ab)}_{sp}$ the projection operator on the proper
symmetry subspace for particles within the cluster $\mathfrak{C}^a_s \cap
\mathfrak{C}^b_p$ then by the HVZ theorem \cite{reed,teschl,beattie} 
\begin{equation}\label{infsigm}
    \inf \sigma \bigl( H_{ab}(\lambda) \mathcal{P}^{(ab)} \bigr) \geq E_{thr}
(\lambda),
\end{equation}
where we define
\begin{equation}\label{pabproj}
    \mathcal{P}^{(ab)} := \mathcal{P}^{(ab)}_{11} \mathcal{P}^{(ab)}_{12}
\mathcal{P}^{(ab)}_{21} \mathcal{P}^{(ab)}_{22} . 
\end{equation}
Note that $[J_a J_b , \mathcal{P}^{(ab)}] = 0$.

\begin{lemma}\label{lemma:3}
Suppose that $H(\lambda)$ satisfies $R1-R4$. If  $\psi_n \wwto \phi_0 $ then
\begin{gather}
    \lim_{n \to \infty} \Bigl\| \left[1- P^{(a)}_{thr}(\lambda_n)\right] J^2_a
(\psi_n - \phi_0) \Bigr\| = 0 \quad (a=1, \ldots, \mathfrak{N}) , \label{claim1}\\
\lim_{n \to \infty} \Bigl\| J^2_a (\psi_n - \phi_0) \Bigr\| = 0  \quad
(a=\mathfrak{N}+1,
\ldots, 2^{N-1}-1) . \label{claim1:1}
\end{gather}
\end{lemma}
\begin{proof}
Note that $\phi_0 \in D(H_0)$ by Lemmas~1, 2(a) in Ref.~\onlinecite{1}. For every $g \in
L^2 (\mathbb{R}^{3N-3}; \mathbb{C}^{n_s})$ we have 
\begin{gather}
 \bigl( g, [1-\mathcal{P}] \phi_0 \bigr) = \bigl([1-\mathcal{P}]  g, \phi_0
\bigr) \nonumber \\
= \lim_{n \to \infty} \bigl([1-\mathcal{P}]  g, \psi_{n} \bigr) = \lim_{n \to
\infty} \bigl( g, [1-\mathcal{P}] \psi_{n} \bigr) = 0. \label{13.07;1}
\end{gather}
Because $g$ in (\ref{13.07;1}) is arbitrary we conclude that $\mathcal{P} \phi_0
= \phi_0$. Consequently, $\mathcal{P}^{(a)} \phi_0 = \phi_0$.

Following the arguments of the proof of Lemma~10 in Ref.~\onlinecite{1} we get
\begin{equation}\label{k112:}
\lim_{n \to \infty} \bigl( (\psi_n-\phi_0), \left[ H(\lambda_n) -
E_{thr}(\lambda_n)
\right](\psi_n-\phi_0) \bigr) = 0.
\end{equation}
We define $\tilde K$ exactly as $K(\lambda)$ in (\ref{14k}), except that all
$V_{ij}$
entering $K$ via $I_a (\lambda)$ and $I_{ab} (\lambda)$ are replaced with
\begin{equation}\label{tildevij}
    \tilde V_{ij} := \tilde U_{ij} + \frac{q_0}{|x_i - x_j|} ,
\end{equation}
where $\tilde U_{ij} := \tilde U (x_i - x_j)$.
Then $\tilde K$ does not depend on $\lambda$ and is relatively $H_0$ compact.
Besides for all $f \in D(H_0)$ one has $|(f, K(\lambda)f)| \leq (f, \tilde K
f)$. Thus
\begin{equation}\label{k112:2}
\lim_{n \to \infty}  \bigl( (\psi_n-\phi_0), K(\lambda_n) (\psi_n-\phi_0) \bigr)
= 0,
\end{equation}
see the proof of Lemma~10 in Ref.~\onlinecite{1}.
Substituting (\ref{ims3}) into (\ref{k112:}) and using (\ref{k112:2}) yields
\begin{gather}\label{19}
\lim_{n \to \infty}  \sum_a \Bigl( (\psi_n - \phi_0), J^2_a \left[H_a
(\lambda_n) -
E_{thr} (\lambda_n)\right]J^2_a (\psi_n - \phi_0)\Bigr) \nonumber \\
+ \sum_{a\neq b} \Bigl( (\psi_n - \phi_0) , J_a J_b \left[H_{ab} (\lambda_n) -
E_{thr}
(\lambda_n)\right] J_b J_a (\psi_n - \phi_0)\Bigr) = 0 . 
\end{gather}
The scalar products under the first sum are clearly non--negative. The terms
under the second sum are non--negative 
by (\ref{infsigm}) (one can insert
$\mathcal{P}^{(ab)}$ because $\mathcal{P}^{(ab)} \mathcal{P} = \mathcal{P}$ and
$[J_a J_b , \mathcal{P}^{(ab)}] = 0$).
Thus we obtain
\begin{equation}\label{allparts}
\lim_{n \to \infty}  \Bigl( (\psi_n - \phi_0), J^2_a \left[H_a (\lambda_n) -
E_{thr} (\lambda_n)\right] J^2_a (\psi_n - \phi_0)\Bigr) = 0
\end{equation}
for all partitions $a=1,\ldots, 2^{N-1}-1$.
Using (\ref{cycha}) and $-\Delta_{R_a}$ being non--negative gives
\begin{equation}\label{allparts:1}
\lim_{n \to \infty}   \Bigl( (\psi_n - \phi_0), J^2_a \left[H^{(a)}_{thr}
(\lambda_n) - E_{thr} (\lambda_n)\right] \mathcal{P}^{(a)} J^2_a (\psi_n -
\phi_0)\Bigr) =
0,
\end{equation}
where we have inserted $\mathcal{P}^{(a)}$. For $a \geq \mathfrak{N}+1$  the
statement of
the lemma given by (\ref{claim1:1}) easily follows from (\ref{allparts:1}) and
R3.
To prove (\ref{claim1}) it suffices to insert into (\ref{allparts:1})
the identity $1 = P^{(a)}_{thr} + (1-P^{(a)}_{thr})$ and to use the inequality
\begin{equation}\label{insineq}
    \left[H^{(a)}_{thr} (\lambda_n) - E_{thr} (\lambda_n)\right]
\left(1-P^{(a)}_{thr}(\lambda_n)
\right)\mathcal{P}^{(a)}  \geq |\Delta
\epsilon|\left(1-P^{(a)}_{thr}(\lambda_n)
\right)\mathcal{P}^{(a)} \quad(a = 1, \ldots, \mathfrak{N}),
\end{equation}
which follows from (\ref{projopa}). \end{proof}
Using Lemma~\ref{lem:thet} we prove
\begin{lemma}\label{lem:6}
Suppose that $H(\lambda)$ satisfies $R1-R4$ and $\psi_n$ defined in $R1$
converges weakly. Then for $a=1,2, \ldots, \mathfrak{N}$ the
sequence $P^{(a)}_{thr} (\lambda_n) \psi_n $ does not spread.
\end{lemma}
\begin{proof}
So let us assume that $\psi_n \wwto \phi_0$, where $\phi_0 \in  D(H_0)$, see
Lemma~1 in Ref.~\onlinecite{1}. The Schr\"odinger equation can be written as
\begin{gather}
    \Bigl\{ \bigl[ H^{(a)}_{thr}(\lambda_n) - E_{thr}(\lambda_n)\bigr ]
-\Delta_{R_a}  +
Q^a (\lambda_n) \eta_{-1} (R_a ) +
    k_n^2 \Bigr \} \psi_n \nonumber \\
    = - \left[ I_a (\lambda_n ) - Q^a (\lambda_n) \eta_{-1} (R_a ) \right]
\psi_n , \label{schro:2}
\end{gather}
where $k_n^2 := E_{thr}(\lambda_n) - E(\lambda_n)$. By Lemma~\ref{lem:4} we can
write
\begin{equation}\label{parepr}
    P^{(a)}_{thr} (\lambda_n ) = \sum_{i=1}^\omega E_i^a(\lambda_n )
P_{\varphi^a_i} (\lambda_n ),
\end{equation}
where $P_{\varphi^a_i} (\lambda_n ) = \varphi^a_i (\cdot, \varphi^a_i) \otimes
1$ and $\varphi^a_i (\lambda_n )$ are
orthonormal eigenstates of $H^{(a)}_{thr}(\lambda_n )$ corresponding to
eigenvalues $E_i^a (\lambda_n )$ ($\varphi^a_i (\lambda_n) = 0$, where
necessary).
By (\ref{projopa}) we have $E_i^a (\lambda_n) \in [E_{thr}(\lambda_n) , E_{thr}(\lambda_n) + |\Delta \epsilon|]$. 
Because the sum in (\ref{parepr}) runs over a finite number of terms (see
Lemma~\ref{lem:4}) to prove
the theorem it suffices to show that $P_{\varphi^a_i} (\lambda_n) \psi_n $ does
not
spread. Acting with $P_{\varphi^a_i} (\lambda_n)$ on both sides of
(\ref{schro:2}) results in
\begin{gather}
\left[ -\Delta_{R_a} + Q^a (\lambda_n) \eta_{-1} (R_a) + {k'}_n^2 \right]
P_{\varphi^a_i} (\lambda_n) \psi_n \label{schro:3}\\
= - P_{\varphi^a_i} (\lambda_n)\left[ I_a (\lambda_n)  - Q^a (\lambda_n)
\eta_{-1} (r) \right] \psi_n
\label{schro:4},
\end{gather}
where ${k'}_n:= \bigl[| E_i^a (\lambda_n) - E_{thr}(\lambda_n)| + k_n^2
\bigr]^{1/2}$. Acting with
\begin{equation}\label{gkn'}
    G_{{k}'_n}^c (Q^a (\lambda_n)) := 1  \otimes \left[ -\Delta_{R_a} + Q^a
(\lambda_n) \eta_{-1}
(R_a) + {k'}_n^2 \right]^{-1}
\end{equation}
on both sides of (\ref{schro:3})--(\ref{schro:4}) we get
\begin{equation}\label{schro:5}
P_{\varphi^a_i} (\lambda_n)  \psi_n  = - G_{{k}'_n}^c (Q^a (\lambda_n))
\eta_{-\delta} (R_a)
P_{\varphi^a_i} (\lambda_n) \eta_{\delta} (R_a)
    \left[ I_a (\lambda_n) - Q^a (\lambda_n) \eta_{-1} (R_a) \right] \psi_n ,
\end{equation}
where we have inserted $\eta_\delta \eta_{-\delta} = 1$ ($\delta$ is defined in
$R2$).
Adding and subtracting $\phi_0$ from $\psi_n$ we rewrite (\ref{schro:5}) as
inequality
\begin{gather}
\left| P_{\varphi^a_i} (\lambda_n) \psi_n \right| \nonumber \\
\leq \left|  G_{{k}'_n}^c (Q^a (\lambda_n))\eta_{-\delta} (R_a)  P_{\varphi^a_i}
(\lambda_n)
\eta_{\delta} (R_a) \left[ I_a (\lambda_n) - Q^a (\lambda_n) \eta_{-1} (R_a)
\right] (\psi_n
-\phi_0)\right| \nonumber \\
+ \left|  G_{{k}'_n}^c (Q^a (\lambda_n))  \eta_{-\delta}
(R_a)P_{\varphi^a_i}
(\lambda_n)\eta_{\delta} (R_a) \left[ I_a (\lambda_n) - Q^a (\lambda_n)
\eta_{-1} (R_a) \right]
\phi_0\right| . 
\end{gather}
This can be continued as 
\begin{gather}
\left| P_{\varphi^a_i} (\lambda_n) \psi_n \right| \nonumber\\
\leq G_{k_n}^c (Q_0)  \eta_{-\delta} (R_a) P_{|\varphi^a_i |} (\lambda_n)\Bigl|
\eta_{\delta} (R_a)  \left[ I_a (\lambda_n) - Q^a (\lambda_n) \eta_{-1} (R_a)
\right] (\psi_n
-\phi_0)\Bigr| \nonumber \\
+ G_{k_n}^c (Q_0)  \eta_{-\delta} (R_a) P_{|\varphi^a_i |}  (\lambda_n) \Bigl|
\eta_{\delta} (R_a)  \left[ I_a (\lambda_n) - Q^a (\lambda_n) \eta_{-1} (R_a)
\right] \phi_0\Bigr|,
\end{gather}
where we define $P_{|\varphi^a_i |} := |\varphi^a_i | (\cdot, |\varphi^a_i
|)\otimes 1$ and use $Q^a (\lambda_n) >Q_0$, ${k'}^2_n \geq k_n^2$ (see remark
after Eq.~(\ref{def:Gk})).
Finally, applying Lemma~\ref{lem:thet} we write
\begin{equation}\label{gnhn}
    \left| P_{\varphi^a_i} (\lambda_n) \psi_n \right|  \leq g_n +
P_{|\varphi^a_i |}
(\lambda_n) e^{\beta |x^a|} h_n ,
\end{equation}
where we define
\begin{gather}
g_n := G_{k_n}^c (Q_0)  \eta_{-\delta} (R_a) P_{|\varphi^a_i |} (\lambda_n)
e^{\beta
|x^a|} \Theta_a (x) \bigl|  \psi_n -\phi_0\bigr| ,  \label{gn}\\
h_n := G_{k_n}^c (Q_0)  \eta_{-\delta} (R_a) \Theta_a (x)\bigl|  \phi_0\bigr| . 
\label{hn}
\end{gather}
It remains to prove that both terms on the rhs of (\ref{gnhn}) do not spread.
We have
\begin{gather}
\| g_n \| \leq \bigl\| G_{k_n}^c (Q_0)  \eta_{-\delta} (R_a) \bigr\| \times
\bigl\| P_{|\varphi^a_i |} (\lambda_n) e^{\beta |x^a|} \bigr\| \times \Bigl\|
\Theta_a (x) (\psi_n -\phi_0)\Bigr\| . 
\end{gather}
The first two operator norms are uniformly bounded by Lemma~\ref{lemma:2}
and R4. Note that $\Theta_a (x)$ is relatively
$H_0$ compact by Lemma~\ref{lem:thet}, see Lemma~7.11 in Ref.~\onlinecite{teschl}.
Therefore, the last norm goes to zero by Lemma~2 in 
Ref.~\onlinecite{1}. Hence, $\|g_n\| \to 0$ and $g_n$ does
not spread. By the same reasoning the sequence $|h_n|$ is uniformly
norm--bounded. The kernel $ G_{k_n}^c$ is pointwise dominated by the kernel of
$G_{k_s}^c$ if $k_s \leq k_n$, see remark after Eq.~(\ref{def:Gk}). Using this
fact it is easy to see
that $h_n$
satisfies the conditions of Lemma~\ref{rsd2} and therefore does not spread.
Now the rhs of (\ref{gnhn}) does not spread by Lemma~\ref{lem:5} since
$\| e^{\beta|x^a|} P_{|\varphi^a_i |} (\lambda_n) e^{\beta|x^a|} \| \leq A^2$ by
R4.
\end{proof}

Now we can prove the main theorem.
\begin{proof}[Proof of Theorem~\ref{th:main}]
We shall first prove that $\psi_n$ does not spread. 
To prove this, by Lemma~4 of Ref.~\onlinecite{1} it is sufficient to 
show that every weakly converging subsequence 
of $\psi_n$ converges also in norm. So let us assume that $\psi_{n_k} \wwto
\phi_0$, where $\psi_{n_k} $ is some weakly convergent 
subsequence of $\psi_n$, $\phi_0 \in D(H_0)$ by Lemmas~1, 2(a) in Ref.~\onlinecite{1}. (Let
us remark that 
one weakly converging subsequence always exists due to the Banach--Alaoglu
theorem). Repeating the arguments in the proof of Lemma~\ref{lemma:3} (around
Eq.~\ref{13.07;1}) 
we conclude that  $\mathcal{P} \phi_0 = \phi_0$. 
Our next aim is to
show that $\psi_{n_k}$ does not spread; by Lemmas~1, 3(a) of 
Ref.~\onlinecite{1} this will imply that $\psi_{n_k} \to \phi_0$ in norm. 
The following identity is obvious
\begin{gather}
\psi_{n_k} = \sum_{a=1}^\mathfrak{N} P^{(a)}_{thr} (\lambda_{n_k}) J_a^2
(\psi_{n_k} -\phi_0) + \sum_{a=1}^\mathfrak{N}
\bigl[ 1-P^{(a)}_{thr} (\lambda_{n_k}) \bigr] J_a^2
(\psi_{n_k} -\phi_0)  \nonumber \\
+ \sum_{a=\mathfrak{N}+1}^{2^{N-1}-1}J_a^2
(\psi_{n_k} -\phi_0) + \phi_0 . \label{obvi:1}
\end{gather}
The last two sums on the rhs go to zero in norm by Lemma~\ref{lemma:3}. It
suffices to prove
that each term in the first sum does not spread. Using $\sum_a J_a^2 =1$ we
write
\begin{gather}
    \bigl| P^{(a)}_{thr} (\lambda_{n_k}) J_a^2  (\psi_{n_k} -\phi_0) \bigr| \leq
\bigl|
P^{(a)}_{thr}  (\lambda_{n_k}) (\psi_{n_k} - \phi_0)\bigr| + \sum_{b\neq a}
\bigl|
P^{(a)}_{thr}(\lambda_{n_k}) J_b^2 ( \psi_{n_k} - \phi_0 )\bigr| \nonumber \\
        \leq  \bigl| P^{(a)}_{thr}  (\lambda_{n_k}) \psi_{n_k} \bigr| + \bigl|
P^{(a)}_{thr}(\lambda_{n_k})
\phi_0 \bigr|  +
\sum_{b\neq a} \bigl| P^{(a)}_{thr}(\lambda_{n_k}) J_b^2 ( \psi_{n_k} - \phi_0
)\bigr| . 
\label{claim3:2}
\end{gather}
The first two terms on the rhs of (\ref{claim3:2}) do not spread by
Lemmas~\ref{lem:6},\ref{lem:5} respectively. It remains to show that each term under the sum on
the rhs of (\ref{claim3:2}) goes to zero in norm. Indeed, for $b \neq a$
\begin{gather}
\bigl\| P^{(a)}_{thr}(\lambda_{n_k}) J_b^2 ( \psi_{n_k} - \phi_0 )\bigr\| \leq
\bigl\|
P^{(a)}_{thr}(\lambda_{n_k}) e^{\beta |x^a|} \bigr\| \times \bigl\|e^{-\beta
|x^a|}
J_b^2 ( \psi_{n_k} - \phi_0 )\bigr\| . 
\end{gather}
The operator norm on the rhs is uniformly bounded by $R4$.  The second norm goes
to zero
because $e^{-\beta |x^a|} J_b^2 \in L^\infty_\infty (\mathbb{R}^{3N-3})$ and is
thus
relatively $H_0$ compact (c. f. Lemma~2 in Ref.~\onlinecite{1} and Lemma~7.11 in
Ref.~\onlinecite{teschl}). 
Thus we have proved that $\psi_{n_k} \to \phi_0$ in norm and  $\psi_n$ does not spread. By Theorem~1 in Ref.~\onlinecite{1}
there exists $\psi_{cr} \in D(H_0)$ such that
$H(\lambda_{cr}) \psi_{cr} = E_{thr} (\lambda_{cr}) \psi_{cr}$. From Eqs.
(10)--(11) in Ref.~\onlinecite{1} it is easy to see that we can set $\psi_{cr} = \phi_0$, 
which results in $\|\psi_{cr} \| = 1$ and $\psi_{cr} = \mathcal{P} \psi_{cr}$. 
\end{proof}

\section{Applications}\label{sec:5}

\subsection{Three Coulomb charges with finite masses}\label{sec:5.1}

We consider the Coulomb Hamiltonian of three particles with charges $\{ q_1 ,
q_2, -1 \} $
and masses $\{ m_1 , m_2, m_3 \} $.
We use Jacobi coordinates $\xi = r_3 - r_2$, $R = r_1 - r_2 - s\xi$, where $s =
m_3 / (m_3 + m_2)$. The Hamiltonian reads \cite{jmpold}
\begin{equation}\label{ap1}
    H(q_1, q_2) = -\frac1{2 \mu_{23}} \Delta_\xi -  \frac 1{2 \mu} \Delta_R -
\frac{q_2}{|\xi |} - \frac{q_1}{|(1-s)\xi - R |} + \frac{q_1 q_2}{|a\xi +R|} , 
\end{equation}
where $\mu_{ik} = m_i m_k /(m_i + m_k )$, $\mu = m_1 (m_2 +m_3
)/(m_1 + m_2 + m_3 )$ are reduced masses.
We keep the masses fixed making $ H(q_1, q_2)$ depend on $q_{1,2} \geq 0$. 
By the Kato's theorem \cite{teschl,reed} $H(q_1 , q_2)$ is a self--adjoint
operator acting in $L^2 (\mathbb{R}^6)$ with the domain 
$D(H) = \mathcal{H}^2 (\mathbb{R}^6)$. The particle spins can be neglected here
and in order to apply the previous formalism we simply 
set all particle spins to zero.

The Hamiltonian  
$ H(q_1, q_2)$ 
is called \textit{stable} if $\inf \sigma  H(q_1, q_2) < E_{thr} (q_1 , q_2)$,
where
$ E_{thr} (q_1 , q_2) := \inf \sigma_{ess} H (q_1 , q_2)$.
A typical stability diagram \cite{martin} for
$ H(q_1, q_2)$ is sketched in Fig.~1. The properties of the stability diagram
are
discussed in detail in Ref.~\onlinecite{martin}. We mention some key features of the
stability diagram, for details see Ref.~\onlinecite{martin}. In the
square $\{q_{1,2}|\: 0 < q_{1,2} <1\}$ the Hamiltonian $H(q_1 , q_2)$ is always
stable (due to
long--range attraction between the
bound pair and the third particle). The line of equal energy
thresholds is determined through $\mu_{23} q_{2}^2 = \mu_{13} q_{1}^2 $ and
divides the plane into upper and lower sectors, where the lowest dissociation
threshold corresponds to $\{123\} \to \{23\} + 1$ and $\{123\} \to \{13\} + 2$
respectively. In each sector the stability area is shaped by two arcs,
which form a cusp on
the line of equal energy thresholds, just like in Fig.~1. The arc in the upper
sector starts at $(\mathfrak{q_1},  1)$ and in the lower sector at $(1,
\mathfrak{q_2})$
and both end up on the line of equal thresholds. 
The points $\{q_{1,2}|\: 0 <
q_1 \leq \mathfrak{q}_1, q_2 = 1 \} $ and
$\{q_{1,2}|\: 0 < q_2 \leq \mathfrak{q}_2, q_1 = 1 \} $ correspond to unstable
$H(q_1 , q_2)$ and the
points
$\{q_{1,2}|\: q_1 \geq \mathfrak{q}_1, q_2 = 1 \} \cap \{q_{1,2}|\: \mu_{23}
q_{2}^2 \geq \mu_{13} q_{1}^2 \}$ and
$\{q_{1,2}|\: q_2 > \mathfrak{q}_2, q_1 = 1 \} \cap \{q_{1,2}|\: \mu_{23}
q_{2}^2 \leq \mu_{13} q_{1}^2 \}$ correspond \cite{martin} to stable $H(q_1 ,
q_2)$. 
Suppose that $H(q_1^o , q_2^o)$ is unstable. If $(q_1^o , q_2^o)$ lies in the
upper sector then $H(q_1^o - s_1, q_2^o)$ and $H(q_1^o , q_2^o + s_2)$ are also unstable,
where $s_1 \in [0, q_1^o ]$ and $s_2 \in [0, \infty)$ respectively. If $(q_1^o , q_2^o)$ lies in the
lower sector then $H(q_1^o + s_1, q_2^o)$ and $H(q_1^o , q_2^o - s_2)$ are also unstable,
where $s_1 \in [0, \infty)$ and $s_2 \in [0, q_2^o ]$ respectively. 
Due to 
the so--called ``overheating'' effect \cite{ruskai} for any given  
$H(q_1 , q_2)$ there exist $s, s' \geq 0$ such
that $H(q_1 + s, q_2)$ and $H(q_1 , q_2 + s')$ would be unstable.  

All properties mentioned above are established rigorously, 
except the fact that $\mathfrak{q}_{1,2} \neq 0$. 
Utilizing the analysis in Refs.~\onlinecite{jmpold,carsten} one can prove that $H(q_1, 1)$
is unstable if both of the following inequalities are fulfilled 
\begin{gather}
 q_1^2 < \frac 3{16} \frac{\mu_{23}}{\mu} , \label{09.07;1}\\
\frac{6 \mu}{\mu_{23}} q_1 \left\{1 + \frac{4q_1 \sqrt{\mu}}{\sqrt{3 \mu_{23}} -
4 q_1 \sqrt{\mu}}\right\} < 1 . \label{09.07;2}
\end{gather}
Note, that from (\ref{09.07;1}) it automatically follows that $(q_1 , 1)$ lies
in the upper sector. From (\ref{09.07;1})--(\ref{09.07;2}) it follows that for
$q_1$ small enough 
$H(q_1, 1)$ is unstable, thus $\mathfrak{q}_1 \neq 0$ and
(\ref{09.07;1})--(\ref{09.07;2}) can be used to derive the lower bound on
$\mathfrak{q}_1$. Similarly, by 
interchanging the indices $1 \leftrightarrow 2$ in
(\ref{09.07;1})--(\ref{09.07;2}) one can get the lower bound on $\mathfrak{q}_2
\neq 0$.

\begin{figure}
\begin{center}
\includegraphics[height=.3\textheight]{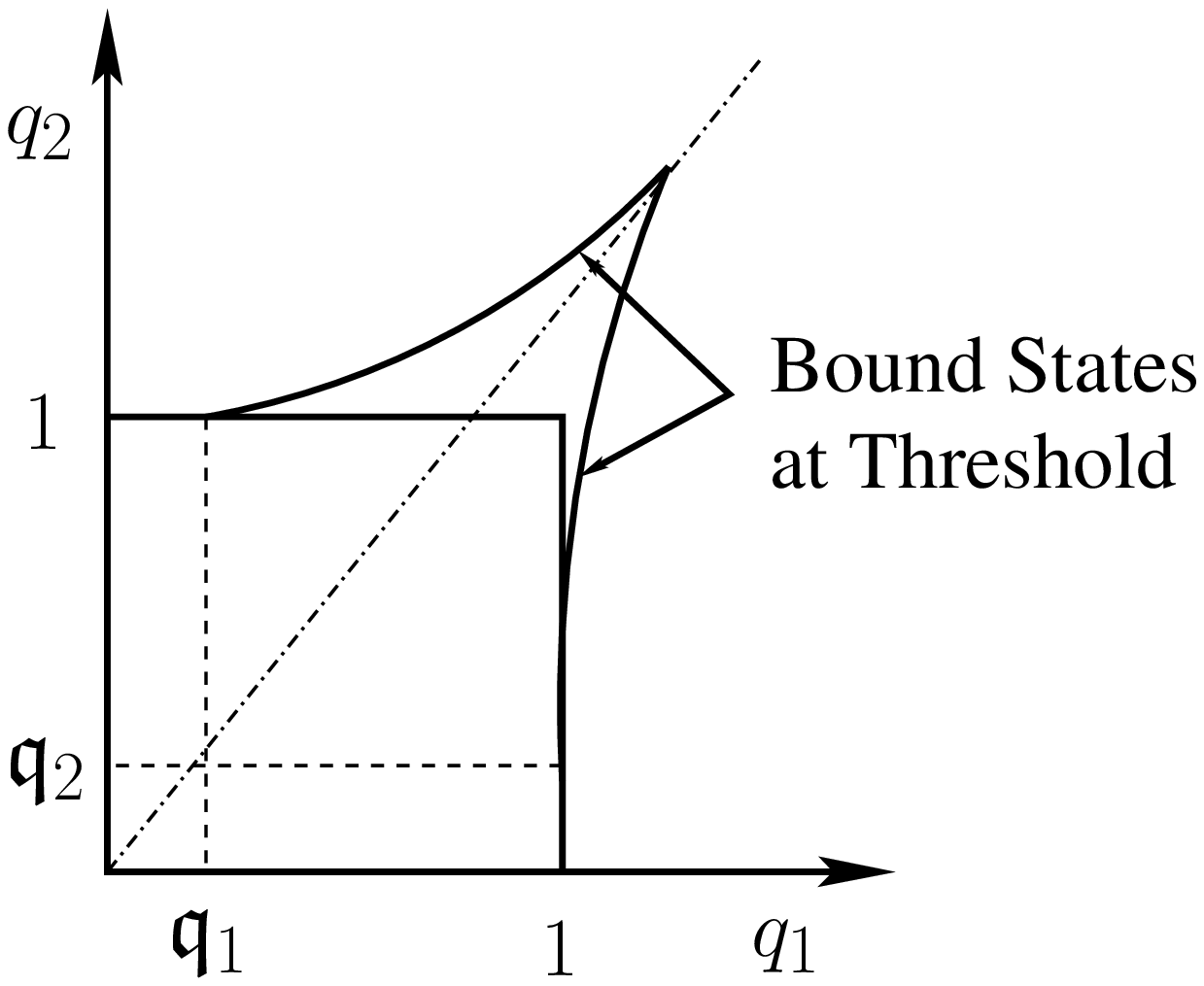}
\caption{The sketch of stability diagram for three Coulomb charges $\{ q_1 ,
q_2 , -1\} $. Systems on the dash--dotted line have equal dissociation
thresholds. The area confined by the unit square and two joint arcs represents
stable systems. On the arcs
of stability curve where either $q_1> 1$ or $q_2 > 1$ there are bound states at
threshold. See also the discussion in Ref.~\onlinecite{3}.}
\end{center}
\end{figure}

Using the results from the previous sections we can prove (see also the
discussion in Sec.~5 in Ref.~\onlinecite{3})
\begin{theorem}\label{th:martin}
Suppose $({\bf q}_1 , {\bf q}_2)$ lies on the stability border in the upper
(resp. lower) sector. (a) If ${\bf q}_2 > 1$ (resp. ${\bf q}_1 > 1$)
then $H({\bf q}_1 , {\bf q}_2)$ has a bound state at threshold.
(b) If ${\bf q}_1 < \mathfrak{q}_1$ (resp. ${\bf q}_2 < \mathfrak{q}_2$) then
$H({\bf q}_1 , {\bf q}_2)$ has no bound states at threshold.
\end{theorem}
\begin{proof}
Let us prove (a). In the vicinity of $({\bf q}_1 , {\bf q}_2)$ one takes a
sequence $(q_{1,2} (\lambda_n) \to {\bf q}_{1,2})$ so that $H\bigl(q_1
(\lambda_n) , q_2 (\lambda_n)\bigr)$ is stable.
For the sequence $\psi_n$ in R1 we take the normalized ground states of
$H\bigl(q_1 (\lambda_n) , q_2 (\lambda_n)\bigr)$. It is straightforward to check
that
all conditions of Theorem~\ref{th:main} can be satisfied. (The requirement R4
can be easily
checked since the exact expressions for the ground state wave functions of the
particle pairs
$\{1,3\}$ and $\{2,3\}$ are known).

Now let us prove (b). Suppose that $({\bf q}_1 , 1)$ lies on the stability
border  in the upper sector  and ${\bf q}_1 < \mathfrak{q}_1$. Assume by
contradiction that there is a normalized $\phi \in D(H)$ such that $H({\bf q}_1
, 1) \phi = E_{thr} ({\bf q}_1 , 1) \phi$
Let us rewrite (\ref{ap1}) for $q_2 = 1$ as
\begin{gather}
     H(q_1, 1) = H_{thr} -  \frac 1{2 \mu} \Delta_R   + W  , \label{cfdgrgrg}\\
     H_{thr} := -\frac1{2 \mu_{23}} \Delta_\xi - \frac 1{|\xi |} , \\
W(q_1) := - \frac{q_1}{|(1-s)\xi - R |} + \frac{q_1 }{|a\xi +R|} . 
\end{gather}
In the upper sector $E_{thr} (q_1, 1) = E_0$ is constant and $H_{thr} \geq E_0$.
Using that $-(\phi, \Delta_R , \phi) > 0$ we get from (\ref{cfdgrgrg})
that $(\phi, W({\bf q}_1)\phi) < 0$. Thus $(\phi, H ({\bf q}_1 + \varepsilon ,
1)\phi) < E_{thr}$, where $\varepsilon = ( \mathfrak{q}_1 - {\bf q}_1)/2 $.
Therefore,
$H ({\bf q}_1 + \varepsilon , 1) = H (\mathfrak{q}_1 - \varepsilon , 1)$ is
stable, which contradicts the properties of the stability border.
\end{proof}
\begin{remark}
Suppose the conditions of Theorem~\ref{th:martin} are fulfilled. From Fig.~1 one
can see that it is possible to construct a sequence of 
points, which correspond to stable Hamiltonians and converge to $({\bf q}_1 , {\bf q}_2)$ (in the topology of $\mathbb{R}^2$). In the case (a) of
Theorem~\ref{th:martin} the 
ground states of these Hamiltonians would form a sequence that does not spread.
In the case (b) the ground states would form a totally spreading sequence \cite{1}. Case (b) bears some similarity to the proof\cite{l2osten} of the 
absence of an $L^2$ -eigenfunction at the bottom of the spectrum of the Hamiltonian of the hydrogen negative ion in the triplet S-sector. 
\end{remark}

\subsection{Negative Atomic Ions}\label{sec:5.2}

We consider the Hamiltonian of an atomic nucleus with charge $Z$ and $N_e$ electrons
\begin{gather}
    H(Z, N_e) = H_0 - \sum_{i=1}^{N_e} \frac Z{|r_i|} + \sum_{1 \leq i <j \leq
N_e}\frac 1{|r_i - r_j|} , \label{atham}\\
H_0 = -\sum_{i=1}^{N_e} \Delta_i - \frac 1M \sum_{1 \leq i <j \leq
N_e} \nabla_i \cdot \nabla_j ,
\end{gather}
where the coordinate $r_i $ points from the nucleus to the electron $i$.
The total number of particles is $N_e +1$ (the electrons are numbered from 1 to
$N_e$
and the nucleus is the particle number $N_e+1$). We set $\hbar = 1$, $m_i = 1$,
$m_{N_e +1} = M$. In the notations of (\ref{xc31})--(\ref{:xc31}) $\lambda = Z$
is the continuous
parameter, $q_i (Z) = -1$ for $i = 1, \ldots, N_e$ and $q_{N_e +1} = Z$. 
The electrons are treated as spin $1/2$ fermions, the spin of the nucleus is set
to zero. 
By $\mathcal{P}_{N_e}$ we shall denote the
projection operator on the subspace of functions, which are antisymmetric with
respect to the interchange of electrons' spin and spatial coordinates. 
The Hamiltonian (\ref{atham}) acts in $L^2 (\mathbb{R}^{3N_e};
\mathbb{C}^{2^{N_e}})$ 
and is a self--adjoint operator with domain $D(H_0) = \mathcal{H}^2 (\mathbb{R}^{3N_e};
\mathbb{C}^{2^{N_e}})$. 
We define 
\begin{equation}
 E(Z, N_e) := \inf \sigma \bigl( H(Z, N_e) \mathcal{P}_{N_e} \bigr) . 
\end{equation}
The nuclear charge $Z_{cr}$ is called \textit{critical} if $E(Z_{cr}, N_e) =
E(Z_{cr}, N_e - 1)$ and $E(Z, N_e) < E(Z, N_e - 1)$ for $Z > Z_{cr}$. It is
known \cite{zhislin} that $Z_{cr} \leq N_e -1$ (due to the long--range
attraction
between the outer electron and remaining particles).
For a rigorous proof on existence of the critical charge see Refs.~\onlinecite{ruskai,sigalsays1,sigalsays2}.
Lieb \cite{lieby1} showed that $Z_{cr} \geq N_e /2$, and in Ref.~\onlinecite{lieby2} one
finds the proof that $Z_{cr} /N_e \to 1$ if $N_e \to \infty$ (here one also
assumes that the nucleus is infinitely heavy). It is generally conjectured that
$Z_{cr} \in (N_e - 2 , N_e -1]$ throughout the periodic system, see, in
particular, Refs.~\onlinecite{hogreve,simonproblems}. Some experimental and theoretically
estimated values of critical charge can be found   in Ref.~\onlinecite{kais}.
Here we prove
\begin{theorem}\label{th:phys}
Suppose that $Z_{cr} \in (N_e - 2 , N_e -1)$. Then there exists $\psi_0 \in
D(H_0)$ , $\|\psi_0 \| = 1$, such that
$H(Z_{cr}, N_e)\psi_0 = E(Z_{cr}, N_e - 1) \psi_0$ and $\mathcal{P}_{N_e} \psi_0
= \psi_0$.
\end{theorem}
\begin{proof}
We need to show that the conditions of Theorem~1 are fulfilled. Since $Z_{cr} >
N_e - 2 $ by assumption, it is known \cite{zhislin} that there exists
$\varepsilon > 0$
such that $E(Z_{cr}, N_e - i) < E(Z_{cr}, N_e -i - 1) - \varepsilon$ for $i = 1,
\ldots, N_e -1$. Due to the continuous dependence of the energies on $Z$ there
exist
$z_0 > 0$ and
$|\Delta \epsilon| \in (0, 2  \varepsilon)$ 
such that for all $Z_n = Z_{cr} + z_0 /n $, where $n = 1,2,\ldots$ one has $Z_n
\in (N_e -2 , N_e -1)$ and 
\begin{equation}\label{qqma}
 E(Z_n, N_e - i) < E(Z_n, N_e -i - 1) - 2  |\Delta \epsilon| \quad \quad (i = 1,
\ldots, N_e -1).
\end{equation}
The requirement R1 is fulfilled if for $\psi_n$ we choose the normalized ground
state of $H(Z_n , N_e)$ that is $H(Z_n , N_e) \psi_n = E(Z_n , N_e) \psi_n$,
where, clearly,
$\mathcal{P}_{N_e} \psi_n = \psi_n$. By the HVZ theorem $E_{thr} (Z) = E(Z, N_e
- 1)$ for $Z= Z_n , Z_{cr}$ and $\psi_n$ exists. R2 is obvious.
The partitions $a = 1, 2, \ldots, \mathfrak{N}$, where $\mathfrak{N} = N_e$,
correspond
to dividing all particles into the electron number $a$ and the rest 
particles. R3 follows from (\ref{qqma}) and the HVZ theorem.
Inequality
(\ref{Q0new}) holds. R4 follows from Lemma~\ref{lem:ahlww}.
\end{proof}
The following Lemma is essentially the result of Ahlrichs \cite{ahlrichs}
generalized to the nucleus of finite mass (see also Ref.~\onlinecite{osten} for a short and
clear exposition).
\begin{lemma}\label{lem:ahlww}
Suppose $H(Z , N_e) \psi = \mathcal{E}\psi$, where $\psi \in D(H_0)$,
$\|\psi\|=1$,
$\mathcal{P}_{N_e} \psi = \psi$ and $\mathcal{E} < E(Z, N_e-1) - |\Delta
\epsilon|$ for some $|\Delta \epsilon| >0$. Then
$\| e^{(4CN_e)^{-1}|r|} \psi\| \leq \sqrt 2$, where $|r| := \sum_i |r_i|$ and
\begin{equation}\label{ahlr5}
    C := \frac{Z}{2|\Delta \epsilon|} + \frac{1}{2|\Delta \epsilon|} \left(
Z^2 + 2|\Delta \epsilon| \right)^{1/2} . 
\end{equation}
\end{lemma}
\begin{proof}
Let $\mathcal{P}'_{N_e}$ denote the projection operator on the subspace of
functions, which are antisymmetric with respect 
to the interchange of spin and spatial coordinates of the electrons $\{2, \ldots
, N_e\}$. Looking at (\ref{atham}) it is easy to
see that
\begin{equation}\label{bseedde}
 H(Z , N_e) \mathcal{P}'_{N_e} + \frac{Z}{|r_1|} \mathcal{P}'_{N_e} \geq E(Z
,N_e -1) \mathcal{P}'_{N_e}
\end{equation}
since in the Hamiltonian on the lhs the first electron is involved only in
positive
interaction terms. From (\ref{bseedde}) it follows that
\begin{equation}
 \Bigl( g, \Bigl\{[ H(Z , N_e) - \mathcal{E}] + Z |r_1|^{-1} \Bigr\} g \Bigr) 
\geq
|\Delta \epsilon| (g,g) , 
\end{equation}
where  $\mathcal{P}'_{N_e} g = g$. Setting $g = f(r_1) \psi(r_1, \ldots,
r_{N_e}, \sigma_1, \ldots, \sigma_{N_e})$ we get
\begin{equation}\label{hermiti}
 \bigl( f \psi, [ H(Z , N_e) - \mathcal{E}] f \psi\bigr) + Z (\psi, |f|^2
|r_1|^{-1} \psi)
\geq  |\Delta \epsilon| (\psi,|f|^2\psi) . 
\end{equation}
Using hermiticity of $i \nabla_j$ one shows \cite{osten} that
\begin{gather}
 \bigl( f \psi, [ H(Z , N_e) - \mathcal{E}] f \psi\bigr) \nonumber \\ 
= \bigl(f\psi, [- \frac
12\Delta_1 , f]
\psi \bigr) - \frac 1M \sum_{j =2}^{N_e} \bigl(f\psi, [\nabla_1 \cdot \nabla_j ,
f] \psi
\bigr)   = \frac 12 \bigl( \psi, |\nabla_1 f|^2 \psi\bigr)  \label{hermit}
\end{gather}
because $ \bigl(f\psi, [\nabla_1 \cdot \nabla_j , f] \psi
\bigr) = 0$ for all $j \geq 2$. Substituting (\ref{hermit}) into (\ref{hermiti})
we produce exactly the inequality (2.6) from Ref.~\onlinecite{osten}. 	
So we can use the inequality (2.20) from Ref.~\onlinecite{osten}, which in our notations
reads
\begin{equation}\label{ahlr3}
    \frac{(\psi, |r_1|^{n+1} \psi)}{(\psi, |r_1
|^n \psi)} \leq \frac 1{2|\Delta \epsilon|} \Bigl\{ Z + \bigl[Z^2 +
\frac{|\Delta \epsilon| }2(n+2)^2 \bigr]^{1/2} \Bigr\} ,
\end{equation}
This can be transformed
into
\begin{equation}\label{ahlr4}
    (\psi, |r_1 |^{n+1} \psi) \leq (n+1)C(\psi, |r_1 |^n \psi) , 
\end{equation}
where $C$ is defined in (\ref{ahlr5}). Since $\|\psi\| = 1$ (\ref{ahlr4})
results in
\begin{equation}\label{ahlr6}
 (\psi, |r_i |^n \psi) \leq C^n n! \quad \quad (i = 1, \ldots, N_e).
\end{equation}
Now using \cite{ahlrichs}
\begin{equation}
 |r|^n \equiv \Bigl( \sum_{i=1}^{N_e} |r_i| \Bigr)^n \leq (N_e)^{n-1}
\sum_{i=1}^{N_e} |r_i|^n
\end{equation}
together with (\ref{ahlr6}) we obtain
\begin{equation}\label{prolem}
 (\psi, |r|^n  \psi) \leq  (C N_e)^n n!
\end{equation}
Using that $\|e^{\beta |r|} \psi\|^2 = \sum_n (2\beta)^n (n!)^{-1} (\psi, |r|^n
\psi)$ and (\ref{prolem}) we prove the Lemma.
\end{proof}

A few remarks are in order. For $N_e = 2$ the statement of Theorem~\ref{th:phys}
was
conjectured in Ref.~\onlinecite{stil} (see also \cite{baker}) and proved in
Ref.~\onlinecite{ostenhof}. For $N_e \geq 3$ this result 
was conjectured in Ref.~\onlinecite{hogreve}. The
restriction $Z_{cr} \in (N_e - 2 , N_e -1)$ in the condition of
Theorem~\ref{th:phys}
is imposed in order to keep the proof completely rigorous (otherwise to apply
Theorem~\ref{th:main} one would need additional assumptions concerning the
nature of dissociation thresholds). In fact, the same result must hold for
$Z_{cr} < N_e -2$.

\appendix

\section{Criteria for Non--Spreading Sequences}\label{sec:6}

Recall \cite{1} that the sequence of functions $f_n
(x) \in L^2
(\mathbb{R}^d)$
\textbf{spreads} if there is $a>0$ such that
\begin{equation}\label{grspre}
 \limsup_{n \to \infty} \|  \chi_{\{x||x| > R\}} f_n \| > a \quad \quad
\textrm{for all $R>0$}.
\end{equation}
(This definition can be found in the papers of Zhislin \cite{zhislin}, who used
the
idea of spreading sequences in his proof of the celebrated HVZ theorem). From
the definition it trivially
follows: (a) if the sequence goes to zero in norm it does not spread; (b)
if $|f_n(x)| \leq \sum_{k=1}^\mathcal{N} g^{(k)}_n (x)$, where each $g^{(k)}_n
(x) \in L^2
(\mathbb{R}^d)$ does not spread and $\mathcal{N}$ is finite, then $f_n$ does not
spread.

\begin{lemma}\label{rsd1}
Suppose the sequence $f_n \in L^2 (\mathbb{R}^d)$ is uniformly norm-bounded and
$|f_n (x)| \leq |f_{n+1} (x)|$. Then $f_n$ does not spread.
\end{lemma}
\begin{proof}
Let us assume by contradiction that $f_n$ spreads, so that (\ref{grspre}) holds.
Let us fix $n$ and choose $R$ so that $\| \chi_{\{x||x| >
R\}} f_n \|^2 < a^2 /4$. Because the sequence $f_n$ spreads we can
find $n'
> n$ such that $\| \chi_{\{x||x| > R\}} f_{n'} \|^2 > a^2 /2$.
Using that $|f_n |$ is non--decreasing we obtain
\begin{gather}
    \| f_{n'} \|^2 = \| \chi_{\{x||x| \leq R\}} f_{n'} \|^2 +  \|  \chi_{\{x||x|
> R\}} f_{n'} \|^2 \geq  \| \chi_{\{x||x| \leq R\}} f_{n} \|^2 +  \|
\chi_{\{x||x| > R\}} f_{n'} \|^2 \nonumber \\
= \| f_{n} \|^2 -  \| \chi_{\{x||x| > R\}} f_{n} \|^2 + \|
\chi_{\{x||x| > R\}} f_{n'} \|^2 \geq     \| f_{n} \|^2  +
\frac{a^2}4 . \label{0.03}
\end{gather}
Thus for any $f_n$ there exists such $f_{n'}$ with $n' >n$ such that $\| f_{n'}
\|^2 \geq \| f_n
\|^2  + a^2 /4$. But this contradicts $f_n$ being a
norm-bounded sequence.\end{proof}

Here is a stronger version of Lemma~\ref{rsd1}.
\begin{lemma}\label{rsd2}
Suppose the sequence $f_n \in L^2 (\mathbb{R}^d)$ is uniformly norm-bounded and
from any subsequence $ f_{n_k}$ one can extract a sub/subsequence $f_{n_{k_s}}$
such that $|f_{n_{k_s}} (x)| \leq |f_{n_{k_{s+1}}}(x)|$.
 Then $f_n$ does not spread.
\end{lemma}
\begin{proof}
Again, let us assume by contradiction that $f_n$ spreads. It follows that for $k
= 1,2,\ldots$ and some $a>0$ one can extract a subsequence
$f_{n_k}$ that satisfies $\| \chi_{\{x| |x| \geq k \}} f_{n_k} \| > a$.
On one hand, it is easy to see that every subsequence of $f_{n_k}$ spreads.
On the other hand, by condition of the Lemma $f_{n_k}$ contains a subsequence,
which is non--decreasing and uniformly bounded, and thus cannot spread by
Lemma~\ref{rsd1},
a contradiction. \end{proof}

 We also need the following 
\begin{lemma}\label{lem:5}
Suppose that $N \geq 3$ and a sequence $f_n \in L^2(\mathbb{R}^{3N-6}; \mathbb{C}^{n_s}) \otimes
L^2(\mathbb{R}^{3})$ is uniformly norm--bounded and does not spread. Suppose
additionally that an operator sequence $A_n \colon L^2(\mathbb{R}^{3N-6}; \mathbb{C}^{n_s}) \to
L^2(\mathbb{R}^{3N-6}; \mathbb{C}^{n_s}) $ is such that $\sup_n \| e^{\alpha |x^a|} A_n \| < K$,
where $K, \alpha > 0$ are constants. Then the sequence $(A_n \otimes 1 )f_n $
does not spread.
\end{lemma}
\begin{proof}
The full set of relative coordinates for a given cluster partition is $x = (x^a,
R_a)$ and
$|x| := |x^a| + |R_a|$. For any given $\varepsilon >0$ let us choose $R >0$ so that
the following inequalities
hold
\begin{gather}
\sup_{x^a} \left[ \chi_{\{ x^a |\; |x^a| \geq R\}} e^{-\alpha |x^a|} \right]  <
\varepsilon/(2K\sup_n\|f_n\|) \label{r1}, \\
\left\| \chi_{\{ R_a |\; |R_a| \geq R \}} f_n \right\| < \varepsilon/2 .
\label{r2}
\end{gather}
Note that
\begin{equation}\label{chigl}
    \chi_{\{ x|\; |x| \geq 2 R\}} \leq \chi_{\{ x^a |\; |x^a| \geq R\}}\otimes 1
 +
1\otimes \chi_{\{ R_a |\; |R_a| \geq R\}} . 
\end{equation}
Using (\ref{chigl}) and (\ref{r1})--(\ref{r2}) we obtain
\begin{equation}\label{chigl2}
    \Bigl\| \chi_{\{ x|\; |x| \geq 2 R\}} (A_n \otimes 1) f_n \Bigr\| <
\varepsilon
\end{equation}
for all $n$. \end{proof}
Obviously, Lemma~\ref{lem:5} also holds if we replace $e^{\alpha |x^a|} $ with
$(1+|x^a|)^\alpha $, where $\alpha >0$ is some power. 

\end{document}